\documentclass[a4paper]{article}

\usepackage{amssymb,amsthm,amsmath}

\usepackage{graphicx}
\usepackage{xcolor}

\usepackage{hyperref}

\usepackage{authblk}

\usepackage{enumitem}

\newtheorem{theorem}{Theorem}[section]

\newtheorem{Lem}[theorem]{Lemma}
\newtheorem{corollary}[theorem]{Corollary}
\newtheorem{Prop}[theorem]{Proposition}

\theoremstyle{definition}

\newtheorem{Def}[theorem]{Definition}

\theoremstyle{remark}

\newtheorem{Rem}[theorem]{Remark}

\numberwithin{equation}{section}

\newcommand{\calA}{\mathcal{A}}

\newcommand{\calH}{\mathcal{H}}

\newcommand{\Hom}{\mathrm{Hom}}

\newcommand{\calB}{{\mathcal B}}
\newcommand{\caH}{{\mathcal H}}

\newcommand{\calP}{{\mathcal P}}
\newcommand{\calU}{{\mathcal U}}

\newcommand{\br}{{\mathbb R}}
\newcommand{\bc}{{\mathbb C}}

\newcommand{\todo}[1]{{\color{red} To Do : #1}}

\newcommand{\sysga}[1]{system of on-site unitary #1 actions}

\newcommand{\id}{\mathop{\mathrm{id}}\nolimits}
\newcommand{\Ad}{\mathop{\mathrm{Ad}}\nolimits}

\newcommand{\Aut}{{\mathrm{Aut}}}
\newcommand{\End}{{\mathrm{End}}}

\newcommand{\norm}[1]{\left\lVert#1\right\rVert}

\newcommand{\xbullet}{{}}
\newcommand{\cxbullet}{{}}

\newif\ifuc
\ucfalse

\begin{document}
\title{On Symmetry-Compatible Superselection Structures for Product States in 2D Quantum Spin Systems}

\author[1]{Matthew Corbelli}%\thanks{email:mdcorbelli@gmail.com}}
\affil[1]{Work done while at the Department of Mathematics, University of California, Davis}

% date of submission to arXiv
\date{October 27, 2025}

\maketitle

\begin{abstract}
We study superselection sectors in two-dimensional quantum spin systems with an on-site action of a compact abelian group $G$. Naaijkens and Ogata (2022) showed that for states quasi-equivalent to a product state, the superselection structure is trivial, reflecting the absence of long-range entanglement. We consider a symmetry-compatible refinement of this setting, in which both the superselection criterion and the notion of equivalence between representations are required to respect the $G$-action. Under this stricter notion of equivalence, the sector structure for a $G$-equivariant product representation becomes nontrivial: the $G$-equivariant superselection sectors are classified by elements of the Pontryagin dual $\widehat{G}$. This shows that even in phases without long-range entanglement, imposing symmetry compatibility can lead to nontrivial sector structure.
\end{abstract}

\section{Introduction}

Superselection theory provides an operator-algebraic framework for classifying physically distinct types of excitations in quantum many-body systems, originally developed in the relativistic setting by Doplicher, Haag, and Roberts \cite{twistedHaagDualityCitation1,twistedHaagDualityCitation2}.
Analogous ideas have since been adapted to quantum spin systems on lattices, such as in the study of anyonic excitations and topological order in gapped phases~\cite{Naaijkens2011LocalizedEndomorphisms,Naaijkens2012HaagDualityToricCode,FielderNaaijkensHaagDualityKitaevQuantumDouble2015,completeSetOfInfntVolGrndStateForKitaevAbelianQDble,ogata2021derivation}.
%\hmm{Within this setting, the structure of superselection sectors captures the presence or absence of long-range entanglement, and thus connects directly to the characterization of topological order}.

In this setting, Naaijkens, Ogata, and collaborators have applied DHR-style analysis to anyonic excitations in quantum spin systems with gapped ground states~\cite{Naaijkens2011LocalizedEndomorphisms,Naaijkens2012HaagDualityToricCode,StablilityOfChargesInInfntQntmSpnSysCNN,NaaijkensOgataApproxSplit2022,ogata2021derivation}.
Ogata~\cite{ogata2021derivation} developed a general construction of the braided $C^*$-tensor category of superselection sectors for pure gapped ground states satisfying approximate Haag duality.
Naaijkens and Ogata~\cite{NaaijkensOgataApproxSplit2022} showed that in two-dimensional systems whose ground states are quasi-equivalent to a product state, the superselection structure is trivial.
They further showed that a state related to such a ground state by a path of automorphisms arising from quasi-local dynamics satisfying suitable conditions, will also have a trivial superselection structure.
By doing so, they showed that the absence of long-range entanglement implies a trivial superselection structure.

%\hmm{Building on this framework}, \hmm{Ogata and collaborators} have developed a systematic operator-algebraic treatment of quantum spin systems with gapped ground states~\cite{ogata2021derivation,NaaijkensOgataApproxSplit2022}. In particular, Naaijkens and Ogata~\cite{NaaijkensOgataApproxSplit2022} showed that in two-dimensional systems whose ground states are quasi-equivalent to a product state, the superselection structure is trivial. \hmm{As part of that analysis, they established a factorization property of quasi-local automorphisms, which implies that states in such topologically trivial phases exhibit no nontrivial sectors.}

In the present work we consider an alteration of this setting that incorporates on-site symmetries. We restrict to systems with an on-site action of a compact abelian group~$G$, and require both the superselection criterion and the notion of equivalence between representations to be compatible with this symmetry. Under this stricter notion of equivalence, the sector structure for a $G$-equivariant product representation is no longer trivial: the sectors are labeled by elements of the Pontryagin dual~$\widehat{G}$ of the symmetry group. This shows that even in phases without long-range entanglement, imposing a symmetry-compatible notion of equivalence can lead to nontrivial superselection structure.
Related analyses of on-site symmetries and associated indices, including classifications of symmetric states in the trivial phase, can be found in~\cite{SymIndex_Ogata_2021}.

\emph{Acknowledgments.} This work was carried out during the author's graduate studies at the University of California, Davis, whose support is gratefully acknowledged. The author was supported in part by the NSF under grant DMS-2407290. The author also thanks Martin Fraas for valuable guidance and discussions.

\section{Preliminaries}\label{sec:prelim}
Let us first state the setting and define the main objects.

We assume the reader has some familiarity with the operator algebraic formulation of quantum spin systems (see e.g.~\cite{Bratteli_Robinson_1987}).

Let $(\Gamma, d:\Gamma \times \Gamma \to \br)$ be a discrete metric space.
For the purposes here, let $(\Gamma,d)$ be a Delone set in $\br^2$, especially a lattice.
For each $x \in \Gamma$, let $\calH_{\{x\}}$ be a finite-dimensional Hilbert space.
Suppose that there is a common upper bound on the dimension of these Hilbert spaces, i.e. that 
\begin{align}\label{agg}
\sup_{x\in\Gamma} \dim \caH_{x}<\infty.
\end{align}

\begin{Def}\label{systemOfOnSiteUnitaryGActions}
    For any topological group $G$, a \emph{\sysga{$G$}} consists of a collection of group homomorphisms $((g \in G) \mapsto (U_{\{x\},g}\in \calU(\calH_{\{x\}}))_{x \in \Gamma}$.
\end{Def}

Let $\calP_{0}(\Gamma)$ denote the set of all finite subsets of $\Gamma$.

For all $\Lambda \in \calP_{0}(\Gamma)$, define $$\calA_{\Lambda} := \bigotimes_{x \in \Lambda}\calB(\calH_{\{x\}})$$
For finite subsets $\Lambda_1 \subseteq \Lambda_2$ of $\Gamma$, there is a natural inclusion of $\calA_{\Lambda_1}$ into $\calA_{\Lambda_2}$. For $\Lambda \subseteq \Gamma$, define $$\calA_{\Lambda}^{\rm loc} := \bigcup_{X \in \calP_0(\Lambda)}\calA_X$$
and define $\calA_\Lambda$ to be the norm-closure of $\calA_{\Lambda}^{\rm loc}$.
Set $\calA := \calA_\Gamma$ and $\calA^{\rm loc} := \calA_{\Gamma}^{\rm loc}$.
%Let $\calA$ denote $\calA_\Gamma$, the algebra of quasi-local observables, and $\calA^{\rm loc}$ denote $\calA_{\Gamma}^{\rm loc}$, the algebra of local observables.

For any topological group $G$ and any \sysga{$G$} $(U_{\{x\}\cxbullet }= ((g \in G) \mapsto (U_{\{x\},g}\in \calU(\calH_{\{x\}})))_{x \in \Gamma}$ , and any $\Lambda \in \calP_{0}(\Gamma)$, define
\[
U_{\Lambda,g} := \prod_{x \in \Lambda} U_{\{x\},g}, \quad \alpha_{\Lambda,g} := \Ad(U_{\Lambda,g}).
\]
Then% these are group homomorphisms from $G$ into $\calU_\Lambda$ (the group of unitaries in $\calA_\Lambda$) and $\Aut(\calA_\Lambda)$, respectively:
\[
(g \mapsto U_{\Lambda,g}) \in \Hom(G, \calU_\Lambda), \quad (g \mapsto \alpha_{\Lambda,g}) \in \Hom(G, \Aut(\calA_\Lambda)).
\]

For infinite $\Lambda \subseteq \Gamma$, define the action on $A \in \calA_{\Lambda}$, and more generally on $A \in \calA$, by the limit
\[
\alpha_{\Lambda,g}(A) := \lim_{\substack{\Lambda' \nearrow \Lambda \\ \text{finite}}} \alpha_{\Lambda',g}(A).
\]
%This limit exists for $A \in \calA_{\Lambda, \text{loc}}$, and more generally in $\calA_{\text{loc}}$, because the sequence is eventually constant. From this, by density and continuity, the definition extends to all of $\calA_{\Lambda}$, and more generally to all of $\calA$.

\subsection{\texorpdfstring{$G$}{G}-covariant representations}

This subsection recalls the theory of $G$-covariant representations and introduces a $G$-equivariant version of the superselection criterion.

\begin{Def}
A state $\omega : \calA \to \mathbb{C}$ is called \emph{$G$-invariant} if
\[
\omega \circ \alpha_g = \omega \quad \text{for all } g \in G.
\]
\end{Def}

The following lemma is well-known:
\begin{Lem}\label{GNSRepOfGInvStateIsGCompatible}
Let $\omega$ be a $G$-invariant state, and let $(\calH_\omega, \pi_\omega, \Omega_\omega)$ be a GNS representation of $\omega$. Then there exists a continuous group homomorphism
\[
U^{(\pi_\omega)} = (g \mapsto U^{(\pi_\omega)}_g) : G \to \calU(\calH_\omega)
\]
such that for all $g \in G$,
\[
\pi_\omega \circ \alpha_g = \Ad(U^{(\pi_\omega)}_g) \circ \pi_\omega, \qquad U^{(\pi_\omega)}_g \Omega_\omega = \Omega_\omega.
\]
%So, $(\pi,U^{(\pi)})$ is a $G$-covariant representation of $(\calA,\alpha)$.
\end{Lem}
See e.g. \cite[Cor.~2.3.17 and §4.3.1]{Bratteli_Robinson_1987}.
%banana [Bratteli–Robinson, Cor. 2.3.17 and §4.3.1 \todo{this should use an actual citation using the bibliography rather than just describing the citation in text}].
%This is a standard result; see, e.g., [Bratteli \& Robinson, Operator Algebras and Quantum Statistical Mechanics I, Corollary 2.3.17 and the start of section 4.3.1 \todo{this should use an actual citation using the bibliography rather than just describing the citation in text}].

We now recall the standard notion of a $G$-covariant representation of a $C^*$-dynamical system.
\begin{Def}
Let $\beta : G \to \Aut(\mathcal{A})$ be a continuous group action.
A \emph{$G$-covariant representation} of $(\mathcal{A}, \beta)$ is a pair $(\pi, U^{(\pi)}_\xbullet)$ of a representation $\pi: \mathcal{A} \to \mathcal{B}(\mathcal{H}_\pi)$ and a strongly continuous unitary representation $U^{(\pi)}_\xbullet: G \to \mathcal{U}(\mathcal{H}_\pi)$ such that for all $g \in G$,
\[
\pi \circ \beta_g = \Ad(U^{(\pi)}_g) \circ \pi.
\]
A $G$-covariant representation $(\pi,U^{(\pi)}_\xbullet)$ of $(\calA,\beta)$ is called \emph{irreducible} if $\pi$ is irreducible.
\end{Def}

\begin{Rem}
By Lemma~\ref{GNSRepOfGInvStateIsGCompatible}, the GNS representation $(\pi_\omega,U^{(\pi_\omega)})$ associated with a $G$-invariant state $\omega$ is $G$-covariant.
\end{Rem}

%\todo{Justify the appropriateness of the following definition}

\begin{Def}\label{def:GEquivariantMap}
    For two $G$-covariant representations $(\pi_1, U^{(\pi_1)}_\xbullet), ( \pi_2, U^{(\pi_2)}_\xbullet)$ of $(\calA,\beta:G \to \Aut(\calA))$, a bounded linear map $T : \calH_{\pi_1}\to \calH_{\pi_2}$ is a \emph{$G$-equivariant map} if
    $\forall A \in \calA, T \pi_1(A) = \pi_2(A) T$ and $\forall g \in G, T U^{(\pi_1)}_g = U^{(\pi_2)}_g T$.
\end{Def}

\begin{Def}\label{SuperSelectionCriterion}
    Let $\pi$ be a representation of $\calA$ serving as a reference representation.
    Another representation $\rho$ is said to satisfy the \emph{superselection criterion} with respect to $\pi$ if, for all cones $\Lambda$, there exists a unitary $V_{\rho,\Lambda} : \calH_\rho \to \calH_\pi$ such that $\Ad(V_{\rho,\Lambda})\circ \rho |_{\calA_{\Lambda^c}} = \pi|_{\calA_{\Lambda^c}}$.
\end{Def}

We now define a version of the superselection criterion suitable for $G$-covariant representations:

\begin{Def}\label{GEquiSSC}
    %With our $\calA$ and the on-site action $g\mapsto \alpha_g$, and a $G$-covariant representation $(\calH_\pi,\pi:\calA \to \calB(\calH_\pi),U^{(\pi)}_\xbullet:G \to \calU(\calH_\pi))$ to serve as the reference representation
    Let $(\pi, U^{(\pi)}_\xbullet)$ be a $G$-covariant representation of $(\calA,\alpha)$ serving as the reference representation.
    Another $G$-covariant representation 
    $(\rho, U^{(\rho)}_\xbullet)$
    is said to satisfy the \emph{$G$-equivariant superselection criterion} with respect to $(\pi,U^{(\pi)}_\xbullet)$ if, for all cones $\Lambda$, there exists a unitary $V_{\rho,\Lambda} : \calH_\rho \to \calH_\pi$ that is $G$-equivariant (i.e. $V_{\rho,\Lambda} U^{(\rho)}_g = U^{(\pi)}_g V_{\rho,\Lambda}$ for all $g \in G$) and such that $\Ad(V_{\rho,\Lambda})\circ \rho |_{\calA_{\Lambda^c}} = \pi|_{\calA_{\Lambda^c}}$.
\end{Def}

\begin{Rem}
The $G$-equivariance condition expresses that the local identifications between sectors commute with the global symmetry, as one expects for physically realizable transformations.
\end{Rem}

\iffalse{
\subsection{Sector theory}\label{sec:select}
\todo{motivation stuff and citing previous work}

\todo{motivation regarding how to define category and cones and stuff about symmetries of the model, and mention DKR theory, etc.}

\todo{motivation of superselection criterion (maybe also mention string operators?)}

\todo{*maybe* mention why cones?}

\todo{mention braided tensor category under assumptions}

\todo{Say something about split property}

 pure state $\omega = \omega_\Lambda \otimes \omega_{\Lambda^c}$ product state with respect to a cone $\Lambda$ and its complement.
 in this case $\pi_\omega(\calA_\Lambda)''$  Type I factor , inclusion $\pi_\omega(\calA_\Lambda)'' \subset \pi_\omega(\calA_{\Lambda^c})'$  split.

\todo{say something about what will prove later}
\todo{should I try and also prove something about if state $\omega$ such that $\omega \circ \alpha$ quasi-equivalent to a product state, for some quasi-factorizable automorphism, then result still works?}

\todo{something about distal/approximate split for }

 $\Lambda_1 \subset \Lambda_2$
tensor product decomposition of the ground state Hilbert space
 $\calA_{\Lambda_1}$ and $\calA_{\Lambda_2^c}$ act on the distinct factors.

}\fi

\section{\texorpdfstring{$\widehat{G}$-grading}{G-hat grading}}\label{sec:GHatGrading}
Let $G$ be a compact abelian group, and let $\widehat{G}$ denote its Pontryagin dual, i.e. the group of continuous group homomorphisms $G \to U(1)$.
We show that the natural $G$-actions -- on local Hilbert spaces $\calH_\Lambda$, local algebras $\calA_\Lambda$, the global algebra $\calA$, Hilbert spaces $\calH_\pi$ of $G$-covariant representations, and the corresponding operator algebras $\calB(\calH_\pi)$ -- induce compatible $\widehat{G}$-gradings on each of these spaces. 
The projections onto the graded components arise as Fourier coefficients of the group action, viewed as a function $G \to \End(B)$ for the relevant Banach space $B$.

\begin{Def}\label{prjcDef}
    Let $B$ be one of $\calH_\Lambda, \calA_\Lambda, \calA, \calH_\pi, \calB(\calH_\pi)$, $G$ be a compact abelian group, and $f : G \to \Aut(B)\subseteq \End(B)$ the relevant $G$ action, and $\mu$ be the normalized Haar measure on $G$.
    Then, for $\phi \in \widehat{G}$, define
    \[
    P_\phi := \int_{g \in G} f(g) \phi(g^{-1}) \, d\mu(g)
    \]
\end{Def}

(For example, $P_\phi : \calH_\Lambda \to \calH_\Lambda$ is defined as $\int_{g \in G} U_{\Lambda,g} \phi(g^{-1}) \, d\mu(g)$.)

The following lemma is well-known:
\begin{Lem}[Orthogonality of Characters]\label{intOverGOfPhiIdentity}
    For any compact abelian group $G$, for any $\phi \in \widehat{G}$, $\int_{g\in G} \phi(g) d\mu(g) = \delta_{\phi,\hat{1}}$ where $\mu$ is the normalized Haar measure for $G$ and $\hat{1}$ is the identity element of $\widehat{G}$. 
\end{Lem}

\begin{Prop}\label{gradePhiProjectionsOnBanachSpace}
    Let $B$ be one of $\calH_\Lambda, \calA_\Lambda, \calA, \calH_\pi, \calB(\calH_\pi)$, and let $f : G \to \Aut(B)\subseteq \End(B)$ the relevant $G$ action, where $G$ is a compact abelian group, and let $\mu$ be its normalized Haar measure.
    Then, for $\phi \in \widehat{G}$,
    $P_\phi$ as defined in Definition~\ref{prjcDef} is well defined, and satisfies these properties:
    \begin{enumerate}
        \item[(a)] For all $g\in G$, $f(g) \circ P_\phi = \phi(g) \cdot P_\phi$
        \item[(b)] $P_\phi$ is a projection, and for $\phi_1,\phi_2 \in \widehat{G}$, $P_{\phi_1}\circ P_{\phi_2} = \delta_{\phi_1,\phi_2} P_{\phi_1}=\delta_{\phi_1,\phi_2}P_{\phi_2}$ (where $\delta_{\phi_1,\phi_2}$ is the Kronecker delta)
        \item[(c)] For $v \in B$, and $\phi \in \widehat{G}$, if for all $g \in G$, $f(g)(v) = \phi(g) v$, then $P_\phi(v) = v$.
    \end{enumerate}
\end{Prop}

\begin{proof}
    For each $\phi \in \widehat{G}$, $\norm{\phi(g^{-1})f(g)}=1$, and $\phi(g^{-1})f(g)$ is measurable, so the Bochner integral defining $P_\phi$ exists, and $\norm{P_\phi} \le 1$.
    %\begin{align*}
    %    \norm{P_\phi} &= \norm{\int_{g \in G} \phi(g^{-1}) f(g) d\mu(g)}\\
    %    &\le \int_{g\in G}\norm{\phi(g^{-1})f(g)}d\mu(g) = 1.
    %\end{align*}
    
    %\norm{P_\phi} = \norm{\int_{g \in G} \phi(g^{-1}) f(g) d\mu(g)}\le \int_{g\in G}\norm{\phi(g^{-1})f(g)}d\mu(g) = 1$.

    \noindent (a) follows by pulling $f(g)$ inside the integral defining $P_\phi$.
    \iffalse{
    To show (a):\\
    As composition on the left with $f(g)$ is a continuous linear map from $\calB(B)$ to $\calB(B)$,
    \begin{align*}
        f(g) \circ P_\phi &= f(g) \circ \int_{g_2\in G}\phi(g_2^{-1})f(g_2) d\mu(g_2)\\
        &=\int_{g_2 \in G}\phi(g_2^{-1})f(g)
        \circ f(g_2) d\mu(g_2)\\
        &=\int_{g_2 \in G}\phi(g) \phi((g g_2)^{-1}) f(g g_2) d\mu(g_2)\\
        &= \phi(g)\cdot P_\phi. 
    \end{align*}
    }\fi
    
    \noindent (b) follows from pulling $P_{\phi_2}$ inside the integral defining $P_{\phi_1}$, applying (a), and applying character orthogonality (Lemma~\ref{intOverGOfPhiIdentity}).
    \iffalse{
    To show (b):\\
    Composition on the right with $P_\phi$ is also a continuous linear map from $\calB(B)$ to $\calB(B)$, and so
    \begin{align*}
        P_{\phi_1}\circ P_{\phi_2} &= \left(\int_{g \in G} \phi_1(g^{-1}) f(g) d\mu(g)\right) \circ P_{\phi_2}\\
        &= \int_{g\in G}\phi_1(g^{-1}) f(g) \circ P_{\phi_2} d\mu(g)\\
        &=\int_{g\in G}\phi_1(g^{-1})\phi_2(g) P_{\phi_2} d\mu(g)\\
        &= \left(\int_{g\in G}(\phi_1^{-1}\phi_2)(g) d\mu(g)\right) P_{\phi_2}\\
        &= \delta_{\phi_1,\phi_2} P_{\phi_2}
    \end{align*}
    where that last equality is by Lemma \ref{intOverGOfPhiIdentity}.
    }\fi
    
    \noindent (c) is immediate from the definition.
    \iffalse{
    To show (c):\\
    Supposing that $v \in B$ and $\phi \in \widehat{G}$ are such that $\forall g \in G, f(g)(v)=\phi(g) v$,
    \begin{align*}
        P_\phi(v) &= \int_{g \in G} \phi(g^{-1}) f(g)(v) d\mu(g)\\
        &=\int_{g \in G} \phi(g^{-1}) \phi(g) \, v \,d\mu(g)\\
        &= \int_{g \in G}\, 1\, d\mu(g) \; v = v \, .
    \end{align*}
    }\fi
    
\end{proof}
(For example, for $B=\calA$, $f(g) = \alpha_g$, and so for $\phi_1, \phi_2 \in \widehat{G}$ and $P_{\phi_1}, P_{\phi_2} : \calA \to \calA$, and $A \in \calA$, $\alpha_g(P_{\phi_1}(A)) = \phi_1(g) P_{\phi_1}(A)$, and $P_{\phi_1}(P_{\phi_2}(A)) = \delta_{\phi_1,\phi_2} P_{\phi_1}(A)$.)

\begin{Rem}
    When $G$ is finite, it can be seen that $\sum_{\phi \in \widehat{G}} P_\phi = \id$.
    For infinite $G$, analogous statements hold for many of these spaces:
    \begin{enumerate}
        \item If $B$ is finite-dimensional, namely $\calH_\Lambda$ or $\calA_\Lambda$ (for finite $\Lambda$), then this sum has only finitely many non-zero terms \ifuc{and therefore converges unconditionally,}\fi and the equality holds.
        \item For any $X \in \calA^{\rm loc}$, the sum $\sum_{\phi \in \widehat{G}} P_\phi(X) = X$, with this sum also having only finitely many non-zero terms.
        \item For any $v \in \calH_\pi$, $\sum_{\phi \in \widehat{G}} P_\phi(v) = v$, with \ifuc{unconditional}\fi convergence in norm (this will be justified in the next subsection via the Peter-Weyl theorem).
    \end{enumerate}
    Moreover, in the Hilbert spaces $\calH_\Lambda$ and $\calH_\pi$ (and $\calA_\Lambda$ using the Hilbert-Schmidt inner product) the ranges of the different projections are orthogonal.

    %Even when $G$ is not finite, for $B \in \{\calH_\Lambda, \calA_\Lambda, \calA^{\rm loc}\}$ and $v \in B$, the sum $\sum_{\phi \in \widehat{G}} P_\phi(v)$ converges unconditionally to $v$, and there are only finitely many non-zero terms in the sum.
    
    %for $v \in \calH_\Lambda$, $\sum_{\phi \in \widehat{G}} P_\phi(v)$ converges unconditionally to $v$, and for $A \in \calA^{\rm loc}$, $\sum_{\phi \in \widehat{G}} P_\phi(A)$ converges unconditionally to $A$ (in both cases, only finitely many terms of the sum are non-zero).
    %In the case of these projections on $B = \calH_\Lambda, \calH_\pi$, the images of these different projections will be orthogonal. If $\calA_\Lambda$ is regarded as a Hilbert space using the Hilbert-Schmidt inner product, the images of the different $P_\phi$ will also be orthogonal.
\end{Rem}

\begin{Lem}
    Suppose $B$ is one of $\calA_\Lambda, \calA, \calB(\calH_\pi)$ so that $f : G \to \Aut(B)$ is an algebra action.
    Then, if $a, b \in B$ are homogeneous of grades $\phi_1, \phi_2 \in \widehat{G}$ respectively   (i.e. $f(g)(a) = \phi_1(g) a$ and $f(g)(b)=\phi_2(g) b$ for all $g \in G$), then $a b$ is homogeneous of grade $\phi_1 \phi_2$.
\end{Lem}
\begin{proof}
    (Immediate.)
    \iffalse{
     \begin{align*}
        f(g)(ab) = f(g)(A) \cdot f(g)(B) = \phi_1(g) a \cdot \phi_2(g) b = (\phi_1\phi_2)(g)\,ab,
    \end{align*}
    }\fi
\end{proof}

\subsection{\texorpdfstring{$\widehat{G}$-grading for $G$-covariant representations}{G-hat grading for G-covariant representations}}\label{sec:GHatGradingOnHilb}
For any $G$-covariant representations $(\calH_\pi, \pi, U_\xbullet^{(\pi)})$ of $(\mathcal{A}, \alpha)$, $(\calH_\pi, U_\xbullet^{(\pi)})$ is a continuous unitary representation of $G$ on a separable Hilbert space, i.e. a $G$-Hilbert space.

%To define a $\widehat{G}$-grading on the infinite-dimensional Hilbert space $\calH_\pi$ of a $G$-covariant representations $(\calH_\pi, \pi, U_\xbullet^{(\pi)})$ of $(\mathcal{A}, \alpha)$, we begin with standard results from the theory of unitary representations.

    %To obtain the grading on the infinite dimensional Hilbert spaces for the $G$-covariant representations of $(\calA,\alpha)$, we begin with standard results from the theory of unitary representations. %we first turn to standard results about unitary representations of groups.

By the Peter–Weyl theorem \cite{PeterWeylTheoremSource}, for any compact group $G$ and any continuous unitary representation $U : G \to \mathcal{U}(\mathcal{H})$ on a separable Hilbert space $\mathcal{H}$, there is a decomposition of $\mathcal{H}$ into a direct sum of finite-dimensional irreducible subrepresentations. In particular, when $G$ is abelian, each irreducible subrepresentation is one-dimensional. Thus,
\[
\mathcal{H} = \bigoplus_{\phi \in \widehat{G}} \mathcal{H}_\phi,
\quad \text{where } 
\mathcal{H}_\phi := \{ v \in \mathcal{H} \mid \forall g \in G,\; U(g) v = \phi(g) v \}.
\]

For each $\phi \in \widehat{G}$, let $p_\phi$ denote the orthogonal projection onto $\mathcal{H}_\phi$. 
\begin{Lem}[From the Peter-Weyl theorem]\label{PWDecomp}
    These subspaces are mutually orthogonal, and every vector $v \in \mathcal{H}$ can be written as
\[
v = \sum_{\phi \in \widehat{G}} p_\phi v,
\]
with \ifuc{unconditional}\fi norm convergence.
\end{Lem}

These projections are equal to the ones defined on $\calH_\pi$ in Definition~\ref{prjcDef}, so they can be written in integral form as
\[
p_\phi v = \int_{G} \phi(g^{-1}) \, U_g v \, d\mu(g),
\]
where $\mu$ is the normalized Haar measure on $G$.
Here, $p_\phi$ is used to distinguish from the projections $P_\phi : \calB(\calH_\pi) \to \calB(\calH_\pi)$,\[
P_\phi(A) := \int_G \phi(g^{-1}) \, \mathrm{Ad}(U_g)(A) \, d\mu(g)\,.
%= \int_G \phi(g^{-1}) \, U_g A U_g^* \, d\mu(g).
\]

These projections on $\calB(\calH_\pi)$ and on $\calH_\pi$ are compatible in the following sense:
\begin{Lem}\label{gradedComponentOfOpOnGradedComponentOfVec}
     Let $\phi_1,\phi_2 \in \widehat{G}$, $A \in \calB(\calH)$ and $v \in \calH$. Then:
     \[
     P_{\phi_1}(A) \cdot p_{\phi_2} v = p_{\phi_1\phi_2} \, A \, p_{\phi_2} v.
     \]
\end{Lem}
\begin{proof}

%more concise version:
\iffalse{
\begin{align*}
P_{\phi_1}(A) \, p_{\phi_2} v
 &= \int_G \phi_1(g^{-1}) \, \Ad(U_g)(A) \, p_{\phi_2} v \, d\mu(g) \\
 &= \int_G \phi_1(g^{-1}) \, U_g A (U_g^* p_{\phi_2} v) \, d\mu(g) \\
 &= \int_G \phi_1(g^{-1}) \phi_2(g^{-1}) \, U_g A p_{\phi_2} v \, d\mu(g) \\
 &= p_{\phi_1 \phi_2} \, A \, p_{\phi_2} v.
\end{align*}
}\fi

%more clear but less concise version:
\iftrue{
    \begin{align*}
        P_{\phi_1}(A) \cdot p_{\phi_2} v
        &= \int_{g \in G}\phi_1(g^{-1}) \Ad(U_g)(A) d\mu(g) p_{\phi_2} v\\
        &= \int_G \phi_1(g^{-1}) \, U_g A U_g^* \, p_{\phi_2} v \, d\mu(g) \\
        &= \int_G \phi_1(g^{-1}) \, U_g A \, \underbrace{U_g^* p_{\phi_2} v}_{= \phi_2(g^{-1}) p_{\phi_2} v} \, d\mu(g) \\
        &= \int_G \phi_1(g^{-1}) \phi_2(g^{-1}) \, U_g A p_{\phi_2} v \, d\mu(g) \\
        &= \int_G (\phi_1 \phi_2)(g^{-1}) \, U_g A p_{\phi_2} v \, d\mu(g) \\
        &= p_{\phi_1 \phi_2} \, A \, p_{\phi_2} v.
    \end{align*}
}\fi

\end{proof}

\begin{Def}\label{DefOfGradeForMapsBetweenSpaces}
    For $(\calH_1, U^{(1)} : G \to \calU(\calH_1))$ and $(\calH_2, U^{(2)} : G \to \calU(\calH_2))$ two Hilbert spaces equipped with a continuous unitary $G$-action, a bounded linear map $T : \calH_1 \to \calH_2$ is said to be homogeneous of grade $\phi \in \widehat{G}$ if $U^{(2)}_gT (U^{(1)}_g)^* = \phi(g) T$ for all $g \in G$.
    
    Note that when $(\calH_1, U^{(1)}) = (\calH_2, U^{(2)})$ this is equivalent to $P_\phi(T)=T$.
\end{Def}
\iffalse{
\begin{Rem}
    \todo{Is this remark needed? Maybe cut it.}
    Projecting onto the $\phi$-graded components for these maps between distinct Hilbert spaces equipped with unitary $G$-actions also works the same (mutatis mutandis) as for bounded operators from one such Hilbert space to itself, but we will not need that here.
\end{Rem}
}\fi

\begin{Lem}\label{ProductOfHomogOpsGradesAdd}
    For $i=1,2,3$ let $(\calH_i, U^{(i)} : G \to \calU(\calH_i))$ be Hilbert spaces equipped with a continuous unitary $G$-action.
    Let $T : \calH_1 \to \calH_2$ be a bounded linear map which is homogeneous of grade $\phi \in \widehat{G}$, and let $S : \calH_2 \to \calH_3$ be a bounded linear map which is homogeneous of grade $\psi \in \widehat{G}$.

    Then, $S \circ T : \calH_1 \to \calH_3$ is homogeneous of grade $\psi \phi$.
\end{Lem}
\begin{proof}
    For all $g \in G$,
    \begin{align*}
        U^{(3)}_g S T \cdot (U^{(1)}_g)^* &= (U^{(3)}_g S \cdot  (U^{(2)}_g)^*) \, (U^{(2)}_g T \cdot (U^{(1)}_g)^*)\\
        &= (\psi(g) S) (\phi(g) T)\\
        &= \psi(g) \phi(g) S T = (\psi \phi)(g) S T.
    \end{align*}
\end{proof}

Now specialize to the case where $(\pi : \mathcal{A} \to \mathcal{B}(\mathcal{H}_\pi), U^{(\pi)} : G \to \mathcal{U}(\mathcal{H}_\pi))$ is a $G$-covariant representation of $(\mathcal{A}, \alpha)$, with $U := U^{(\pi)}$.

\begin{Lem}\label{homogVecNotKernelAndNonzeroGradedPartsOfOp}
    Let $A \in \mathcal{B}(\mathcal{H}_\pi)$. If $A \ne 0$, then:
    \begin{enumerate}
    \item There exists $\phi \in \widehat{G}$ such that $P_\phi(A) \ne 0$.
    \item For every such $\phi$, there exists 
    a non-zero homogenous vector $u \in \calH_\pi$ of grade $\phi_1 \in \widehat{G}$ such that $A\,u\ne 0$ and $P_\phi(A) \,u \ne 0$.
    %a non-zero vector $u \in \mathcal{H}_\pi$ and $\phi_1 \in \widehat{G}$ such that:
    %\begin{itemize}
    %    \item $u$ is homogeneous of grade $\phi_1$ (i.e., $u = p_{\phi_1} u$),
    %    \item $A u \ne 0$,
    %    \item $P_\phi(A) u \ne 0$.
    %\end{itemize}
    \end{enumerate}
    %Let $A \in \calB(\calH_\pi)$. Then, if $A\ne 0$, there is some $\phi \in \widehat{G}$ such that $P_\phi(A) \ne 0$. Furthermore, for every $\phi$ such that $P_{\phi}(A) \ne 0$, there is a $u \in \calH$ and a $\phi_1 \in \widehat{G}$ such that $u$ is homogeneous of grade $\phi_1$ (i.e. $u = p_{\phi_1} u$) and $A u \ne 0$ and $P_{\phi}(A) u \ne 0$.
    %%Let $A \in \calB(\calH_\pi)$ and $\phi \in \widehat{G}$. If $P_\phi(A) \ne 0$, there exists a $v \in \calH_\pi$ and a $\phi_1 \in \widehat{G}$ such that $P_{\phi}(A) p_{\phi_1} v \ne 0$. Also, if $A \ne 0$, there exists a $\phi \in \widehat{G}$ such that $P_\phi(A) \ne 0$.
\end{Lem}
\begin{proof}
Let $A \ne 0$. Then there exists $v \in \calH_\pi$ with $Av \ne 0$.
Write $v = \sum_{\phi_1} p_{\phi_1}v$ (by Lemma~\ref{PWDecomp}) and $Av = \sum_{\phi_1} A p_{\phi_1} v$.
As this is non-zero, for some $\phi_1$ we have $Ap_{\phi_1}v \ne 0$, and then for some $\phi_2$ we also have $p_{\phi_2}Ap_{\phi_1}v \ne 0$.
Setting $\phi = \phi_2\phi_1^{-1}$, Lemma~\ref{gradedComponentOfOpOnGradedComponentOfVec} gives
\[
P_\phi(A)p_{\phi_1}v = p_{\phi\phi_1}Ap_{\phi_1}v = p_{\phi_2}Ap_{\phi_1}v \neq 0,
\]
so $P_\phi(A)\ne 0$.

For the second claim, now let $\phi$ be any element of $\widehat{G}$ such that $P_\phi(A) \neq 0$. Apply the above reasoning with $P_\phi(A)$ in place of $A$, and set $u:=p_{\phi_1}v$.

\iffalse{
    Let $A \in \calB(\calH_\pi)$ be non-zero. Then there exists some $v \in \calH_\pi$ such that $A v \ne 0$.

As $\mathcal{H}_\pi$ decomposes as $\bigoplus_{\phi_1 \in \widehat{G}} (\mathcal{H}_\pi)_{\phi_1}$, we can write:
\[
v = \sum_{\phi_1 \in \widehat{G}} p_{\phi_1} v, \quad \text{with unconditional convergence.}
\]
Then,
\[
A v = \sum_{\phi_1 \in \widehat{G}} A p_{\phi_1} v.
\]
    
Since $A v \ne 0$, at least one term in the sum is non-zero, say $A p_{\phi_1} v \ne 0$ for some $\phi_1 \in \widehat{G}$.
    
Now decompose $A p_{\phi_1} v$ into its graded components. There exists $\phi_2 \in \widehat{G}$ such that:
\[
p_{\phi_2} A p_{\phi_1} v \ne 0.
\]
Set $\phi := \phi_2 \phi_1^{-1}$, so that $\phi_2 = \phi \phi_1$.

By Lemma~\ref{gradedComponentOfOpOnGradedComponentOfVec}, we have:
\[
P_\phi(A) \cdot p_{\phi_1} v = p_{\phi \phi_1} A p_{\phi_1} v = p_{\phi_2} A p_{\phi_1} v \ne 0.
\]
Therefore, $P_\phi(A) \ne 0$.

Set $u := p_{\phi_1} v$. Then $u$ is homogeneous of grade $\phi_1$, $A u \ne 0$, and $P_\phi(A) u \ne 0$.

For the second part: let $\phi \in \widehat{G}$ be such that $P_\phi(A) \ne 0$. Then apply the above argument to $P_\phi(A)$ in place of $A$ to find such a homogeneous vector $u$.
}\fi

\end{proof}

%\todo{The following needs to be fixed up}
\begin{Lem}\label{OpHasOnlyOneGradedComponentImpliesHomogeneous} 
    If $A \in \calB(\calH_\pi)$ and there is exactly one $\phi \in \widehat{G}$ such that $P_\phi(A) \ne 0$, then $P_\phi(A)=A$, i.e. $A$ is homogeneous of grade $\phi$.
\end{Lem}
\begin{proof}

Suppose $A$ has exactly one nonzero component, say $P_\phi(A)\ne 0$.  
Let $v \in \calH_\pi$ and $\phi_1 \in \widehat{G}$. By Lemma~\ref{PWDecomp},
\[
A p_{\phi_1} v = \sum_{\phi_2 \in \widehat{G}} p_{\phi_2} A p_{\phi_1} v.
\]
By Lemma~\ref{gradedComponentOfOpOnGradedComponentOfVec}, each term equals
\[
p_{\phi_2} A p_{\phi_1} v = P_{\phi_2\phi_1^{-1}}(A) \, p_{\phi_1} v.
\]
Since $P_\phi(A)$ is the only nonzero component of $A$, the sum reduces to the unique term with $\phi_2 \phi_1^{-1} = \phi$. Thus%, for all $v \in \calH_\pi$ and $\phi_1 \in \widehat{G}$,
\[
A p_{\phi_1} v = P_\phi(A)\, p_{\phi_1} v.
\]

Now let $v \in \calH_\pi$. Again by Lemma~\ref{PWDecomp},
$v = \sum\limits_{\phi_1 \in \widehat{G}} p_{\phi_1} v$,
with norm convergence. Therefore
\[
Av = \sum_{\phi_1} A p_{\phi_1} v 
   = \sum_{\phi_1} P_\phi(A) p_{\phi_1} v
   = P_\phi(A) \sum_{\phi_1} p_{\phi_1} v
   = P_\phi(A) v.
\]
Hence $A = P_\phi(A)$.

\end{proof}

\subsection{Refining the \texorpdfstring{$\widehat{G}$}{G-hat}-grading}\label{SectionRefiningGHatGrading}\label{sec:refiningGHatGrading}

\begin{Def}
Let $((g \in G) \mapsto (U_{\{x\},g}\in \calU(\calH_{\{x\}}))_{x \in \Gamma}$ be a \sysga{$G$}, and let $\Lambda \subset \Gamma$. Then, for each $x \in \Gamma$, define
    \[
    \tilde{U}_{\{x\},(g_1,g_2)} := 
    \begin{cases}
        U_{\{x\},g_1} & \text{if } x \in \Lambda\\
        U_{\{x\},g_2} & \text{if } x \in \Lambda^c.
    \end{cases}
    \]
    So $(\tilde{U}_{\{x\}\cxbullet}: (G \times G) \to \calU(\calH_{\{x\}}))_{x \in \Gamma}$ is a \sysga{$G\times G$}.
    %This defines a system $((g_1,g_2) \mapsto \tilde{U}_{\{x\},(g_1,g_2)})_{x \in \Gamma}$ of on-site unitary $G \times G$-actions.

    %(More generally, this can be extended to any partition of $\Gamma$, with one copy of $G$ for each part. Here, we consider only the case of a region and its complement.)

    From this, for each finite region $\Lambda_1 \in \calP_0(\Gamma)$, define the maps $(g_1,g_2) \mapsto \tilde{U}_{\Lambda_1,(g_1,g_2)}$ and $(g_1,g_2) \mapsto \tilde{\alpha}_{\Lambda_1,(g_1,g_2)}$ just as for any \sysga{group}, and likewise define $(g_1,g_2) \mapsto \tilde{\alpha}_{(g_1,g_2)}$.
   
    Note that for all $g \in G$ and for all finite $\Lambda_1 \in \calP_0(\Gamma)$, that $\tilde{U}_{\Lambda_1,(g,g)}=U_{\Lambda_1,g}$. Also note that for all $g \in G$, $\tilde{\alpha}_{(g,g)}=\alpha_g$.
\end{Def}

The Pontryagin dual of $G \times G$ is (isomorphic to) the group $\widehat{G} \times \widehat{G}$.

With this \sysga{$G\times G$}, an element of $\calA_{\Lambda}$ which has grade $\phi \in \widehat{G}$ with respect to the grading obtained from the \sysga{$G\times G$}, has grade $(\phi,\hat{1}) \in \widehat{G}\times \widehat{G}$ with respect to the grading obtained from this \sysga{$G\times G$}. Likewise, an element of $\calA_{\Lambda^c}$ of grade $\phi\in \widehat{G}$ with respect to the grading from the $G$ action, has grade $(\hat{1},\phi) \in \widehat{G}\times \widehat{G}$ with respect to the grading from this $G \times G$ action.

This is because the $G \times G$ action acts on $\calA_\Lambda$ only by the first copy of $G$, and acts on $\calA_{\Lambda^c}$ only by the second copy of $G$. So, $\tilde{\alpha}_{(g,1)}|_{\calA_\Lambda}=\alpha_g|_{\calA_\Lambda}$ and $\tilde{\alpha}_{(1,g)}|_{\calA_{\Lambda^c}}=\alpha_g|_{\calA_{\Lambda^c}}$.

The purpose of the next few lemmas is largely in order to circumvent issues of convergence that may arise when trying to express operators $A \in \calB(\calH_\pi)$ as sums of their homogeneous components, either with respect to the $\widehat{G}$-grading or the $(\widehat{G}\times \widehat{G})$-grading. (These issues do not arise when $G$ is finite.)

\begin{Lem}\label{GHatProjsCompatWithGHatByGhatProjs}
    Let $(\pi:\calA\to\calB(\calH_\pi),U^{(\pi)}:G \to \calU(\calH_\pi))$ be a $G$-covariant representation of $(\calA,\alpha)$ and $(\pi:\calA\to\calB(\calH_\pi),\tilde{U}^{(\pi)} : G\times G\to\calU(\calH_\pi))$ be a $G\times G$-covariant representation of $(\calA,\tilde{\alpha})$. Suppose that $\forall g \in G, \tilde{U}^{(\pi)}_{(g,g)}=U^{(\pi)}_{g}$. Then, for all $\phi_1,\phi_2 \in \widehat{G}$,
    \[
    p_{\phi_1 \phi_2} p_{(\phi_1,\phi_2)}=p_{(\phi_1,\phi_2)}=p_{(\phi_1,\phi_2)} p_{\phi_1 \phi_2}\,,
    \] where the projections $p_{(\phi_1,\phi_2)}$ for $(\phi_1,\phi_2)\in \widehat{G}\times \widehat{G}$ are the projections onto the $(\phi_1,\phi_2)\in \widehat{G}\times \widehat{G}$ grade components of $\calH_\pi$, defined using $\tilde{U}^{(\pi)}$ just as the projections $p_{\phi}$ are defined using $U^{(\pi)}$.
\end{Lem}
\begin{proof}
    First to show that $p_{\phi_1 \phi_2} p_{(\phi_1,\phi_2)}=p_{(\phi_1,\phi_2)}$:
    \begin{align*}
        p_{\phi_1 \phi_2} p_{(\phi_1,\phi_2)} &= \int_{g\in G}(\phi_1 \phi_2)(g^{-1}) U^{(\pi)}_g d\mu(g) \, p_{(\phi_1,\phi_2)}\\
        %&= \int_{g\in G}(\phi_1 \phi_2)(g^{-1}) U^{(\pi)}_g p_{(\phi_1,\phi_2)} d\mu(g) \, \\
        &= \int_{g\in G}(\phi_1 \phi_2)(g^{-1}) \tilde{U}^{(\pi)}_{(g,g)} p_{(\phi_1,\phi_2)} d\mu(g) \, \\
        &= \int_{g\in G}(\phi_1 \phi_2)(g^{-1}) \,\;(\phi_1,\phi_2)(g,g) \; p_{(\phi_1,\phi_2)} d\mu(g)\\
        %&= \int_{g\in G}(\phi_1 \phi_2)(g^{-1}) \phi_1(g)\phi_2(g)  p_{(\phi_1,\phi_2)} d\mu(g)\\
        &= \int_{g\in G}p_{(\phi_1,\phi_2)}d\mu(g) = p_{(\phi_1,\phi_2)}.
    \end{align*}
    Showing that $p_{(\phi_1,\phi_2)}=p_{(\phi_1,\phi_2)} p_{\phi_1 \phi_2}$ is essentially the same, except that instead of using $\tilde{U}^{(\pi)}_{(g_1,g_2)} p_{(\phi_1,\phi_2)}=(\phi_1,\phi_2)(g_1,g_2) p_{(\phi_1,\phi_2)}$ to conclude that $U^{(\pi)}_g p_{(\phi_1,\phi_2)} = (\phi_1,\phi_2)(g,g) p_{(\phi_1,\phi_2)}$, it uses $p_{(\phi_1,\phi_2)} \tilde{U}^{(\pi)}_{(g_1,g_2)}=(\phi_1,\phi_2)(g_1,g_2) p_{(\phi_1,\phi_2)}$ to conclude that $ p_{(\phi_1,\phi_2)}U^{(\pi)}_g = (\phi_1,\phi_2)(g,g) p_{(\phi_1,\phi_2)}$. (One concludes that $p_{(\phi_1,\phi_2)} \tilde{U}^{(\pi)}_{(g_1,g_2)}=(\phi_1,\phi_2)(g_1,g_2) p_{(\phi_1,\phi_2)}$ from the fact that $\tilde{U}^{(\pi)}_{(g_1,g_2)} p_{(\phi_1,\phi_2)}=(\phi_1,\phi_2)(g_1,g_2) p_{(\phi_1,\phi_2)}$ and $p_{(\phi_1,\phi_2)} \tilde{U}^{(\pi)}_{(g_1,g_2)} v = p_{(\phi_1,\phi_2)} \tilde{U}^{(\pi)}_{(g_1,g_2)} \sum_{(\phi_1',\phi_2')\in \widehat{G}\times \widehat{G}} p_{(\phi_1',\phi_2')}v$ for all $v \in \calH_\pi$.)
\end{proof}

\begin{Lem}\label{nonzeroGHatByGHatCompImpliesNonzeroGHatComp}
    Let $(\pi:\calA\to\calB(\calH_\pi),U^{(\pi)}:G \to \calU(\calH_\pi))$ be a $G$-covariant representation of $(\calA,\alpha)$ and $(\pi:\calA\to\calB(\calH_\pi),\tilde{U}^{(\pi)} : G\times G\to\calU(\calH_\pi))$ be a $G\times G$-covariant representation of $(\calA,\tilde{\alpha})$. Suppose that $\forall g \in G, \tilde{U}^{(\pi)}_{(g,g)}=U^{(\pi)}_{g}$. Let $\phi_1,\phi_2 \in \widehat{G}$ and $A \in \calB(\calH_\pi)$.
    
    Then, if $P_{(\phi_1,\phi_2)}(A)\ne 0$, then $P_{\phi_1 \phi_2}(A) \ne 0$.
\end{Lem}
\begin{proof}
    Suppose $P_{(\phi_1,\phi_2)}(A) \ne 0$.
    Then, by Lemma~\ref{homogVecNotKernelAndNonzeroGradedPartsOfOp} there exists a homogeneous vector $v \in \calH_\pi$ of grade $(\psi_1,\psi_2)\in \widehat{G}\times \widehat{G}$ such that $P_{(\phi_1,\phi_2)}(A) v = P_{(\phi_1,\phi_2)}(A) p_{(\psi_1,\psi_2)}v \ne 0$.

    %Suppose $P_{(\phi_1,\phi_2)}(A) \ne 0$.
    %Then, by Lemma~\ref{homogVecNotKernelAndNonzeroGradedPartsOfOp} there exists a vector $v \in \calH_\pi$ which is homogeneous with respect to the $\widehat{G}\times \widehat{G}$ grading, of some grade $(\psi_1,\psi_2)\in \widehat{G}\times \widehat{G}$, and such that $P_{(\phi_1,\phi_2)}(A) v = P_{(\phi_1,\phi_2)}(A) p_{(\psi_1,\psi_2)}v \ne 0$.

    By Lemma~\ref{gradedComponentOfOpOnGradedComponentOfVec}, $P_{(\phi_1,\phi_2)}(A) p_{(\psi_1,\psi_2)}v = p_{(\phi_1,\phi_2)(\psi_1,\psi_2)} A p_{(\psi_1,\psi_2)} v$.

    By Lemma~\ref{GHatProjsCompatWithGHatByGhatProjs}, $p_{(\psi_1, \psi_2)}v = p_{\psi_1 \psi_2} p_{(\psi_1,\psi_2)}v$, so, (using $v=p_{(\psi_1,\psi_2)}v$) we have $v=p_{\psi_1 \psi_2} v$.

    Therefore,
    \begin{align*}
        p_{(\phi_1,\phi_2)(\psi_1,\psi_2)} P_{\phi_1 \phi_2}(A) v 
        &=p_{(\phi_1,\phi_2)(\psi_1,\psi_2)} P_{\phi_1 \phi_2}(A) p_{\psi_1 \psi_2} v\\
        &=p_{(\phi_1,\phi_2)(\psi_1,\psi_2)}p_{\phi_1\phi_2\psi_1\psi_2}A p_{\psi_1\psi_2}v\\
        &=p_{(\phi_1\psi_1,\phi_2\psi_2)}p_{\phi_1\phi_2\psi_1\psi_2}A p_{\psi_1\psi_2}v\\
        &=p_{(\phi_1\psi_1,\phi_2\psi_2)}A p_{(\psi_1,\psi_2)}v\\
        &=P_{(\phi_1,\phi_2)}(A) p_{(\psi_1,\psi_2)}v \ne 0.
    \end{align*}
    (The fourth equality is using Lemma~\ref{GHatProjsCompatWithGHatByGhatProjs} on the left part of the expression and $p_{\psi_1 \psi_2} v = v = p_{(\psi_1,\psi_2)}v$ on the right part of the expression.)
    
    %So, as $p_{(\phi_1,\phi_2)(\psi_1,\psi_2)} P_{\phi_1 \phi_2}(A) v \ne 0$, therefore $P_{\phi_1 \phi_2}(A) v\ne 0$, and so $P_{\phi_1 \phi_2}(A) \ne 0$.

    So, $p_{(\phi_1,\phi_2)(\psi_1,\psi_2)} P_{\phi_1 \phi_2}(A) v \ne 0$, and therefore $P_{\phi_1 \phi_2}(A) \ne 0$.
\end{proof}

\begin{Lem}\label{GHatByGHatOnDoubleCommutant}
    Let $(\pi:\calA\to\calB(\calH_\pi),U^{(\pi)}:G \to \calU(\calH_\pi))$ be a $G$-covariant representation of $(\calA,\alpha)$ and $(\pi:\calA\to\calB(\calH_\pi),\tilde{U}^{(\pi)} : G\times G\to\calU(\calH_\pi))$ be a $G\times G$-covariant representation of $(\calA,\tilde{\alpha})$. Suppose that $\tilde{U}^{(\pi)}_{(g,g)} = U^{(\pi)}_g$ for all $g \in G$.

    Let $A \in \pi(\calA_\Lambda)''$ and $B \in \pi(\calA_{\Lambda^c})''$.
    \iffalse{Then:
    \begin{itemize}
        \item For all $(g_1, g_2) \in G \times G$,
        \[
        \Ad(\tilde{U}^{(\pi)}_{(g_1,g_2)})(A) = \Ad(U^{(\pi)}_{g_1})(A), \quad \Ad(\tilde{U}^{(\pi)}_{(g_1,g_2)})(B) = \Ad(U^{(\pi)}_{g_2})(B).
        \]
        \item For all $\phi \in \widehat{G}$,
        \[
        P_{(\phi, \hat{1})}(A) = P_\phi(A), \quad P_{(\hat{1}, \phi)}(B) = P_\phi(B).
        \]
    \end{itemize}}\fi
    Then,
\[
\Ad(\tilde{U}^{(\pi)}_{(g_1,g_2)})(A)=\Ad(U^{(\pi)}_{g_1})(A),
\qquad
\Ad(\tilde{U}^{(\pi)}_{(g_1,g_2)})(B)=\Ad(U^{(\pi)}_{g_2})(B),
\]
for all $(g_1,g_2) \in G\times G$.
    
    %For all $(g_1,g_2) \in G\times G$, $\Ad(\tilde{U}^{(\pi)}_{(g_1,g_2)})(A)=\Ad(U^{(\pi)}_{g_1})$ and $\Ad(\tilde{U}^{(\pi)}_{(g_1,g_2)})(B)=\Ad(U^{(\pi)}_{g_2})$.
    
    Furthermore, $P_{(\phi,\hat{1})}(A)=P_\phi(A)$ and $P_{(\hat{1},\phi)}(B)=P_\phi(B)$.
\end{Lem}
\begin{proof}
    The arguments about $A \in \pi(\calA_\Lambda)''$ (involving the first copy of $G$ and $\widehat{G}$) and about $B \in \pi(\calA_{\Lambda^c})''$ (involving the second copies) are symmetric, so we give the proof only for $A$.
    
    Since $\tilde{\alpha}_{(1,g)}$ acts trivially on $\calA_\Lambda$, we have
$\Ad(\tilde{U}^{(\pi)}_{(1,g)})(\pi(a))=\pi(a)$ for all $a\in\calA_\Lambda$, hence $\tilde{U}^{(\pi)}_{(1,g)}\in\pi(\calA_\Lambda)'=(\pi(\calA_\Lambda)'')'$. Thus $\Ad(\tilde{U}^{(\pi)}_{(1,g)})(A)=A$ for all $A\in\pi(\calA_\Lambda)''$ and $g \in G$.

 So, for $(g_1,g_2) \in G\times G$ and $A \in \pi(\calA_\Lambda)''$,
 \iffalse{
 \[
\Ad(\tilde{U}^{(\pi)}_{(g_1,g_2)})(A)
 =\Ad(\tilde{U}^{(\pi)}_{(g_1,g_1)}\tilde{U}^{(\pi)}_{(1,g_1^{-1}g_2)})(A)
 =\Ad(\tilde{U}^{(\pi)}_{(g_1,g_1)})(A)
 =\Ad(U^{(\pi)}_{g_1})(A).
\]
 }\fi
    \begin{align*}
        \Ad(\tilde{U}^{(\pi)}_{(g_1,g_2)})(A)
        &=\Ad(\tilde{U}^{(\pi)}_{(g_1,g_1)}  \tilde{U}^{(\pi)}_{(1,g_1^{-1}g_2)})(A)\\
        &=\Ad(\tilde{U}^{(\pi)}_{(g_1,g_1)})(A)
        =\Ad(U^{(\pi)}_{g_1})(A)\,.
    \end{align*}

    \iffalse{
    Since for all $g \in G$,
    \[
    \Ad(\tilde{U}^{(\pi)}_{(1,g)})(\pi(a)) = \pi(\tilde{\alpha}_{(1,g)}(a)) = \pi(a) \quad \text{for all } a \in \calA_\Lambda,
    \]
    it follows that $\tilde{U}^{(\pi)}_{(1,g)} \in \pi(\calA_\Lambda)' = \left(\pi(\calA_\Lambda)''\right)'$. (Here, $1$ denotes the identity element of $G$.)
    
    Thus, for $g \in G$, and $A \in \pi(\calA_\Lambda)''$, $\Ad(\tilde{U}^{(\pi)}_{(1,g)})(A)=A$.

    So, for $(g_1,g_2) \in G\times G$ and $A \in \pi(\calA_\Lambda)''$,
    \begin{align*}
        \Ad(\tilde{U}^{(\pi)}_{(g_1,g_2)})(A)
        &=\Ad(\tilde{U}^{(\pi)}_{(g_1,g_1)} \cdot \tilde{U}^{(\pi)}_{(1,g_1^{-1}g_2)})(A)\\
        &=\Ad(\tilde{U}^{(\pi)}_{(g_1,g_1)})(A)
        =\Ad(U^{(\pi)}_{g_1})(A)\,.
    \end{align*}
    }\fi

    For the grading projections, we compute, for $A \in \pi(\calA_\Lambda)''$ and $\phi \in \widehat{G}$:
    \begin{align*}
        P_{(\phi, \hat{1})}(A) 
        &= \int_{G \times G}\!\! \overline{(\phi, \hat{1})(g_1, g_2)} \Ad(\tilde{U}^{(\pi)}_{(g_1, g_2)})(A) \, d\mu_{G \times G}(g_1, g_2) \\
        &= \int_{G \times G}\!\! \phi(g_1^{-1}) \Ad(U^{(\pi)}_{g_1})(A) \, d\mu_{G \times G}(g_1, g_2) \\
        &= \int_{g_1 \in G}\!\! \phi(g_1^{-1}) \Ad(U^{(\pi)}_{g_1})(A) \, d\mu_G(g_1) = P_\phi(A),
    \end{align*}
    since the integrand is independent of $g_2$.
    %since the integral over $g_2$ contributes nothing (the integrand is constant in $g_2$).
    
    The argument for $P_{(\hat{1}, \phi)}(B) = P_\phi(B)$ is the same, with the roles of $\Lambda$ and $\Lambda^c$ reversed.
\end{proof}

\begin{Lem}\label{gradeProjectionsPreserveSupport}
Let $(\pi,U^{(\pi)})$ be a $G$-covariant representation of $(\mathcal{A},\alpha)$.
Then, for every region $\Lambda_1 \subseteq \Gamma$ and every $X \in \pi(\calA_{\Lambda_1})''$, one has
\[
P_\phi(X)\in\pi(\mathcal{A}_{\Lambda_1})''\qquad(\phi\in\widehat G).
\]

\iffalse{
banana

    Let $(\pi:\calA\to\calB(\calH_\pi),U^{(\pi)}:G \to \calU(\calH_\pi))$ be a $G$-covariant representation of $(\calA,\alpha)$ and $(\pi:\calA\to\calB(\calH_\pi),\tilde{U}^{(\pi)} : G\times G\to\calU(\calH_\pi))$ be a $G\times G$-covariant representation of $(\calA,\tilde{\alpha})$.
    
    Then, for all $\phi \in \widehat{G}$ and all $(\phi_1,\phi_2) \in \widehat{G}\times \widehat{G}$, for all regions $\Lambda_1 \subseteq \Gamma$, for all $X \in \pi(\calA_{\Lambda_1})''$, $P_{\phi}(X), P_{(\phi_1,\phi_2)}(X) \in \pi(\calA_{\Lambda_1})''$. (In particular one can apply this to $\Lambda_1 = \Lambda$ or to $\Lambda_1 = \Lambda^c$.)
}\fi

\end{Lem}
\begin{proof}
Fix $\Lambda_1$ and $X\in\pi(\mathcal{A}_{\Lambda_1})''$.
Since $\alpha$ is on-site and $U^{(\pi)}$ implements it, we have
\(\Ad(U^{(\pi)}_g)(\pi(\mathcal{A}_{\Lambda_1}))=\pi(\mathcal{A}_{\Lambda_1})\) for all \(g\in G\).
So, for all $g \in G$, $\Ad(U^{(\pi)}_g)(X) \in \Ad(U^{(\pi)}_g)(\pi(\calA_{\Lambda_1})'')=(\Ad(U^{(\pi)}_g)(\pi(\calA_{\Lambda_1})))''=\pi(\mathcal{A}_{\Lambda_1})''$.
By definition
\[
P_\phi(X)=\int_G \phi(g^{-1})\,\Ad(U^{(\pi)}_g)(X)\,d\mu_G(g),
\]
which is a Bochner integral of a bounded \(\pi(\mathcal{A}_{\Lambda_1})''\)-valued function.
Since the von Neumann algebra \(\pi(\mathcal{A}_{\Lambda_1})''\) is closed under such integrals, the integral (hence \(P_\phi(X)\)) lies in \(\pi(\mathcal{A}_{\Lambda_1})''\).

\iffalse{
banana

    Let $\Lambda_1 \subseteq \Gamma$.
    
    As $\alpha$ and $\tilde{\alpha}$ are on-site and $U^{(\pi)}$ and $\tilde{U}^{(\pi)}$ represent them for $\pi$, for all $g \in G$ and all $(g_1,g_2) \in G \times G$, $\Ad(U^{(\pi)}_g)(\pi(\calA_{\Lambda_1}))=\pi(\calA_{\Lambda_1})$ and $\Ad(\tilde{U}^{(\pi)}_{(g_1,g_2)})(\pi(\calA_{\Lambda_1}))=\pi(\calA_{\Lambda_1})$.

    For any set $S$ of operators and any unitary $U$, $(\Ad(U)(S))''=\Ad(U)(S'')$. 

    Therefore, for $X \in \pi(\calA_{\Lambda_1})''$, $\Ad(U^{(\pi)}_g)(X),\Ad(\tilde{U}^{(\pi)}_{(g_1,g_2)})(X) \in \pi(\calA_{\Lambda_1})''$.

    What remains then is that the Bochner integrals defining $P_{\phi}(A)=\int_{g\in G}\phi(g^{-1}) \Ad(U^{(\pi)}_g)(X) d\mu_G(g)$ and $P_{(\phi_1,\phi_2)}(A) = \int_{(g_1,g_2)\in G\times G}(\phi_1,\phi_2)((g_1,g_2)^{-1}) \Ad(\tilde{U}^{(\pi)}_{(g_1,g_2)})(X) d\mu_{G\times G}(g_1,g_2)$ of functions with values in $\pi(\calA_{\Lambda_1})''$, are still in $\pi(\calA_{\Lambda_1})''$.

    Since $\pi(\mathcal{A}_{\Lambda_1})''$ is a von Neumann algebra, it is a Banach space, and so is closed under (converging) Bochner integrals of operator-valued functions taking values in it. So these absolutely convergent Bochner integrals in $\pi(\mathcal{A}_{\Lambda_1})''$ converge to values in $\pi(\mathcal{A}_{\Lambda_1})''$.
    
    I.e. $P_{\phi}(X), P_{(\phi_1,\phi_2)}(X) \in \pi(\calA_{\Lambda_1})''$.
}\fi

\end{proof}

\begin{Rem}
    As this applies to any such compact $G$ action, in particular it applies in the case of $(\pi,\tilde{U}^{(\pi)})$ a $G\times G$-covariant representation of $(\calA,\tilde{\alpha})$, and using $P_{(\phi_1,\phi_2)}$ for $(\phi_1,\phi_2) \in \widehat{G}\times \widehat{G}$.
\end{Rem}

\begin{Lem}\label{prodOfDisjSuppAndIdGradeImplHomogFactors}
    Let $(\pi:\calA\to\calB(\calH_\pi),U^{(\pi)}:G \to \calU(\calH_\pi))$ be a $G$-covariant representation of $(\calA,\alpha)$ and $(\pi:\calA\to\calB(\calH_\pi),\tilde{U}^{(\pi)} : G\times G\to\calU(\calH_\pi))$ be a $G\times G$-covariant representation of $(\calA,\tilde{\alpha})$. Suppose that $\forall g \in G, \tilde{U}^{(\pi)}_{(g,g)}=U^{(\pi)}_{g}$.

    Suppose also that $\pi$ has the property that if $X \in \pi(\calA_\Lambda)''$ and $Y \in \pi(\calA_{\Lambda^c})''$ and $X,Y$ are non-zero, then $XY \ne 0$, where $\Lambda$ is the region in terms of which the $(G \times G)$-action is defined. 
    
    Let $A \in \pi(\calA_\Lambda)''$ and $B \in \pi(\calA_{\Lambda^c})''$.

    Then, if $AB$ is non-zero and is homogeneous of grade $\hat{1} \in \widehat{G}$, then $A$ and $B$ are each homogeneous with respect to the $\widehat{G}$-grading, and their grades are inverses of each-other.
\end{Lem}

\begin{proof}
Since $AB \neq 0$, we have $A \neq 0$ and $B \neq 0$. By Lemma~\ref{homogVecNotKernelAndNonzeroGradedPartsOfOp} the sets
\[
S_A:=\{\phi\in\widehat G:\;P_\phi(A)\neq0\},\qquad
S_B:=\{\phi\in\widehat G:\;P_\phi(B)\neq0\}
\]
are nonempty.

%    By Lemma~\ref{GHatByGHatOnDoubleCommutant},
%    %because the $(G \times G)$-action acts only by the first copy of $G$ on $\pi(\calA_\Lambda)''$ and only by the second copy of $G$ on $\pi(\calA_{\Lambda^c})''$,
%    for such $\phi_A, \phi_B \in \widehat{G}$, $P_{\phi_A}(A) = P_{(\phi_A,\hat{1})}(A)$ and $P_{\phi_B}(B)=P_{(\hat{1},\phi_B)}(B)$.

\iffalse{
By Lemma~\ref{GHatByGHatOnDoubleCommutant},
%the projections coming from the $G\times G$–action restrict on the two region algebras, so
for any $\phi_A,\phi_B\in\widehat{G}$ we have
\[
P_{\phi_A}(A)=P_{(\phi_A,\hat{1})}(A)
\qquad\text{and}\qquad
P_{\phi_B}(B)=P_{(\hat{1},\phi_B)}(B).
\]
}\fi

By Lemma~\ref{GHatByGHatOnDoubleCommutant}, we have that, for all $(g_1,g_2) \in G\times G$,
    \begin{align*}
        \Ad(\tilde{U}^{(\pi)}_{(g_1,g_2)})(AB) &= \Ad(\tilde{U}^{(\pi)}_{(g_1,g_2)})(A) \, \Ad(\tilde{U}^{(\pi)}_{(g_1,g_2)})(B)\\
        &= \Ad(U^{(\pi)}_{g_1})(A) \, \Ad(U^{(\pi)}_{g_2})(B).
    \end{align*}
So, for all $\phi_A, \phi_B \in \widehat{G}$,
    \begin{align*}
        P_{(\phi_A,\phi_B)}(AB)
        &= \int_{(g_1,g_2)\in G\times G}\hspace{-3.8em}\overline{(\phi_A,\phi_B)(g_1,g_2)} \Ad(\tilde{U}^{(\pi)}_{(g_1,g_2)})(AB) d\mu_{G \times G}(g_1,g_2)
        \end{align*}\vspace{-1.7em}\begin{align*}
        \hphantom{P_{(\phi_A,\phi_B)}(AB)}\hspace{-29pt}&=\int_{(g_1,g_2)\in G\times G}\hspace{-4em}\phi_A(g_1^{-1}) \phi_B(g_2^{-1})\Ad(U^{(\pi)}_{g_1})(A) \, \Ad(U^{(\pi)}_{g_2})(B)d\mu_{G \times G}(g_1,g_2)\\
        &=\int_{(g_1,g_2)\in G\times G}\hspace{-3.6em}\phi_A(g_1^{-1}) \Ad(U^{(\pi)}_{g_1})(A) \,\phi_B(g_2^{-1}) \Ad(U^{(\pi)}_{g_2})(B)d\mu_{G \times G}(g_1,g_2)\\
        &=\int_{g_1 \in G}\int_{g_2 \in G}\hspace{-1.3em}\phi_A(g_1^{-1}) \Ad(U^{(\pi)}_{g_1})(A) \,\phi_B(g_2^{-1}) \Ad(U^{(\pi)}_{g_2})(B)d\mu_G(g_2)d\mu_G(g_1)\\
        &=\int_{g_1 \in G}\hspace{-1.3em}\phi_A(g_1^{-1}) \Ad(U^{(\pi)}_{g_1})(A) \int_{g_2 \in G} \hspace{-1.3em}\phi_B(g_2^{-1}) \Ad(U^{(\pi)}_{g_2})(B)d\mu_G(g_2)d\mu_G(g_1)\\
        &=\int_{g_1 \in G}\hspace{-1.3em}\phi_A(g_1^{-1}) \Ad(U^{(\pi)}_{g_1})(A) d\mu_G(g_1)\,\int_{g_2 \in G}\hspace{-1.3em} \phi_B(g_2^{-1}) \Ad(U^{(\pi)}_{g_2})(B)d\mu_G(g_2)\\
        &= P_{\phi_A}(A) P_{\phi_B}(B).
    \end{align*}

By Lemma~\ref{gradeProjectionsPreserveSupport}, $P_{\phi_A}(A)\in\pi(\mathcal A_\Lambda)''$ and $P_{\phi_B}(B)\in\pi(\mathcal A_{\Lambda^c})''$.

Let $\phi_A \in S_A$ and $\phi_B \in S_B$, so that $P_{\phi_A}(A), P_{\phi_B}(B)$ are both nonzero. Then, since both are nonzero and $\pi$ was assumed to have the property that for nonzero operators $X \in \pi(\calA_\Lambda)'', Y \in \pi(\calA_{\Lambda^c})''$,  $XY\ne 0$, we get
\[
P_{(\phi_A,\phi_B)}(AB)=P_{\phi_A}(A)P_{\phi_B}(B)\neq 0.
\]

By Lemma~\ref{nonzeroGHatByGHatCompImpliesNonzeroGHatComp} then, $P_{\phi_A \phi_B}(AB) \ne 0$. So, $$S_{AB} = \{\phi \in \widehat{G} : \; P_\phi(AB) \ne 0\} \supseteq \{\phi_A \phi_B : \; \phi_A \in S_A, \phi_B \in S_B\}.$$

But $AB$ is assumed homogeneous of grade $\hat{1}\in \widehat{G}$, so the only $\phi_A \phi_B$ that can occur is $\hat{1}$.
Therefore, for every $\phi_A \in S_A$ (of which there is at least one), every $\phi_B \in S_B$ must be $\phi_A^{-1}$, so all elements of $S_B$ are equal to each-other, i.e. $S_B$ has only one element, and conversely, all elements $\phi_A \in S_A$ must be the inverse of this element of $S_B$.
Thus, $S_A=\{\phi_A\}$ and $S_B=\{\phi_B\}$ for some $\phi_A, \phi_B \in \widehat{G}$ with $\phi_A \phi_B = \hat{1}$.
Then by Lemma~\ref{OpHasOnlyOneGradedComponentImpliesHomogeneous} we conclude $A=P_{\phi_A}(A)$ and $B=P_{\phi_B}(B)$; i.e. \(A\) and \(B\) are homogeneous of grades \(\phi_A\) and \(\phi_B=\phi_A^{-1}\), as claimed.

\iffalse{
Therefore for every $(\phi_A,\phi_B)\in S_A\times S_B$ we must have
\[
\phi_A \phi_B=\hat{1}\qquad\Longrightarrow\qquad \phi_B=\phi_A^{-1}.
\]

Now take arbitrary $\phi\in S_A$ and $\psi\in S_B$. Using the product formula for the $(\widehat G\times\widehat G)$–projection and the fact that $\Ad(\tilde U^{(\pi)}_{(g_1,g_2)})(AB)=\Ad(U^{(\pi)}_{g_1})(A)\,\Ad(U^{(\pi)}_{g_2})(B)$ (the latter equality follows from Lemma~\ref{GHatByGHatOnDoubleCommutant}), we obtain
\[
P_{(\phi,\psi)}(AB)
= \int_{G\times G}\! \overline{(\phi,\psi)(g_1,g_2)}\,\Ad(\tilde U^{(\pi)}_{(g_1,g_2)})(AB)\,d\mu
= P_\phi(A)\,P_\psi(B),
\]
where the Fubini step is justified since the integrand is bounded and measurable (product Haar measure). The last equality is the product of the two Bochner integrals defining $P_\phi(A)$ and $P_\psi(B)$.
}\fi

\end{proof}

\section{Classification of Superselection sectors with respect to a Product Representation}\label{sec:mainResult}
The following theorem is a specialization of Theorem~4.5 of \cite{NaaijkensOgataApproxSplit2022}. They prove a more general version stated for quasi-equivalence; here we consider the special case with irreducible representations and therefore use unitary equivalence, which suffices for our purposes.
\begin{theorem}[Theorem 4.5 from \cite{NaaijkensOgataApproxSplit2022}]\label{productRepTrivialSectorTheory}
    Let $\Lambda \subset \Gamma$ be a cone. Let $\pi_\Lambda : \calA_{\Lambda}\to \calB(\calH_\Lambda)$ and $\pi_{\Lambda^c} : \calA_{\Lambda^c}\to \calB(\calH_{\Lambda^c})$ be irreducible representations of $\calA_\Lambda$ and $\calA_{\Lambda^c}$ respectively. Let $\pi_0 := \pi_\Lambda \otimes \pi_{\Lambda^c}$.

    Then, if any irreducible representation $\sigma : \calA \to \calB(\calH_\sigma)$ satisfies the superselection criterion (Definition~\ref{SuperSelectionCriterion}) with respect to $\pi_0$, then $\sigma$ is unitarily equivalent to $\pi_0$.%, i.e. there exists a unitary $U : \calH_\sigma \to \calH_\Lambda \otimes \calH_{\Lambda^c}$ such that $\Ad(U)\circ \sigma = \pi_0$.
\end{theorem}

\begin{Rem}\label{RemarkGByGCovariantRepFromProductOfGCovReps}
Given $G$-covariant representations $(\pi_\Lambda, U^{(\pi_\Lambda)})$ and $(\pi_{\Lambda^c}, U^{(\pi_{\Lambda^c})})$ of $(\calA_{\Lambda},\alpha_\Lambda)$ and $(\calA_{\Lambda^c},\alpha_{\Lambda^c})$ respectively, the tensor product $(\pi_0 := \pi_\Lambda \otimes \pi_{\Lambda^c}, U^{(\pi_0)} := (g \mapsto U^{(\pi_\Lambda)}_g \otimes U^{(\pi_{\Lambda^c})}_g)\,)$ is a $G$-covariant representation of $(\calA,\alpha)$.
Moreover, defining $\tilde{U}^{(\pi_0)}_{(g_1, g_2)} := U^{(\pi_\Lambda)}_{g_1} \otimes U^{(\pi_{\Lambda^c})}_{g_2}$, $(\pi_0, \tilde{U}^{(\pi_0)})$ is a $G \times G$-covariant representation of $(\calA,\tilde{\alpha})$ satisfying $\tilde{U}^{(\pi_0)}_{(g,g)} = U^{(\pi_0)}_g$ for all $g$.

\iffalse{
    Let $\Lambda \subset \Gamma$. Let $(\pi_\Lambda, U^{(\pi_\Lambda)})$ and $(\pi_{\Lambda^c}, U^{(\pi_{\Lambda^c})})$ be $G$-covariant representations of $(\calA_{\Lambda},\alpha_\Lambda)$ and $(\calA_{\Lambda^c},\alpha_{\Lambda^c})$ respectively. Set $\pi_0 := \pi_\Lambda \otimes \pi_{\Lambda^c}$ and $U^{(\pi_0)}:= (g \mapsto U^{(\pi_\Lambda)}_g \otimes U^{(\pi_{\Lambda^c})}_g)$. Then $(\pi_0 , U^{(\pi_0)})$ is a $G$-covariant representation of $(\calA,\alpha)$. In addition, if one defines the map $\tilde{U}^{(\pi_0)}:=((g_1, g_2) \mapsto U^{(\pi_\Lambda)}_{g_1} \otimes U^{(\pi_{\Lambda^c})}_{g_2}) : G \times G \to \calU(\calH_{\pi_\Lambda} \otimes \calH_{\pi_{\Lambda^c}})$ then $(\pi_0 , \tilde{U}^{(\pi_0)})$ is a $G\times G$-covariant representation of $(\calA,\tilde{\alpha})$, such that $\forall g \in G, \tilde{U}^{(\pi_0)}_{(g,g)}=U^{(\pi_0)}_g$.
}\fi

\end{Rem}

\begin{theorem}\label{classificationOfSectorsForGEquivariantSSCWrtProductRep}
Let $\Lambda \subset \Gamma$ be a cone.
    Let $(\pi_\Lambda, U^{(\pi_\Lambda)}_\xbullet)$ and $(\pi_{\Lambda^c}, U^{(\pi_{\Lambda^c})}_\xbullet)$ be irreducible $G$-covariant representations of $(\calA_{\Lambda},\alpha_{\Lambda})$ and $(\calA_{\Lambda^c},\alpha_{\Lambda^c})$ respectively.
    %, with $\pi_{\Lambda}: \calA_{\Lambda} \to \calB(\calH_{\Lambda})$ and $\pi_{\Lambda^c} : \calA_{\Lambda^c} \to \calB(\calH_{\Lambda^c})$ being irreducible representations.
    Let $(\pi_0, U^{(\pi_0)}_\xbullet)$ be the $G$-covariant representation of $(\calA,\alpha)$ obtained as $\pi_0 := \pi_\Lambda \otimes \pi_{\Lambda^c}$ and $U^{(\pi_0)}_g := U^{(\pi_\Lambda)}_g \otimes U^{(\pi_{\Lambda^c})}_g$.

    Let $(\sigma : \calA \to \calB(\calH_\sigma), U^{(\sigma)}_\xbullet)$ be an irreducible $G$-covariant representation of $\calA$ which satisfies the $G$-equivariant superselection criterion (Definition~\ref{GEquiSSC}) with respect to $(\pi_0, U^{(\pi_0)}_\xbullet)$.% and let $\sigma$ be an irreducible representation.

    Then, there exists a unique $\phi \in \widehat{G}$ such that there is a unitary $U : \calH_\sigma \to \calH_{\pi_0}=\calH_{\Lambda}\otimes \calH_{\Lambda^c}$ of grade $\phi$ (in the sense defined in Definition~\ref{DefOfGradeForMapsBetweenSpaces}) such that $\Ad(U)\circ \sigma =\pi_0$. (There are also no non-homogeneous $U : \calH_\sigma \to \calH_{\pi_0}$ that satisfy $\Ad(U)\circ \sigma =\pi_0$.)

    In this sense, the irreducible $G$-covariant representations of $(\calA,\alpha)$ satisfying the $G$-equivariant version of the superselection criterion with respect to $(\pi_0, U^{(\pi_0)})$ are classified by $\widehat{G}$ up to $G$-equivariant unitary equivalence.
\end{theorem}
\begin{proof}
    As $(\sigma,U^{(\sigma)}_\xbullet)$ satisfies the $G$-symmetry respecting version of the superselection criterion with respect to $(\pi_0,U^{(\pi_0)}_\xbullet)$, $\sigma$ satisfies the superselection criterion (the version not dealing with a $G$-action, Definition ~\ref{SuperSelectionCriterion}) with respect to $\pi_0$.
    
    Therefore, by Theorem \ref{productRepTrivialSectorTheory}, there exists a unitary $U : \calH_\sigma \to \calH_{\pi_0}$ such that $\Ad(U) \circ \sigma = \pi_0$.

    Because $(\sigma,U^{(\sigma)}_\xbullet)$ satisfies the the $G$-equivariant version of the superselection criterion (Definition~\ref{GEquiSSC}) with respect to $(\pi_0,U^{(\pi_0)}_\xbullet)$, there exist $G$-equivariant unitaries $V_{\sigma,\Lambda}, V_{\sigma,\Lambda^c} : \calH_\sigma \to \calH_{\pi_0}$ such that $\Ad(V_{\sigma,\Lambda})\circ \sigma |_{\calA_{\Lambda^c}}=\pi_0|_{\calA_{\Lambda^c}}$ and $\Ad(V_{\sigma,\Lambda^c})\circ \sigma|_{\calA_\Lambda}=\pi_0|_{\calA_\Lambda}$.
    So $\Ad(V_{\sigma,\Lambda}^*)\circ \pi_0|_{\calA_{\Lambda^c}}=\sigma|_{\calA_{\Lambda^c}}$ and $\Ad(V_{\sigma,\Lambda^c}^*)\circ \pi_0|_{\calA_{\Lambda}}=\sigma|_{\calA_{\Lambda}}$, and therefore $\Ad(U V_{\sigma,\Lambda}^*)\circ \pi_0|_{\calA_{\Lambda^c}}=\Ad(U) \circ \sigma|_{\calA_{\Lambda^c}}=\pi_0|_{\calA_{\Lambda^c}}$ and $\Ad(U V_{\sigma,\Lambda^c}^*)\circ \pi_0|_{\calA_{\Lambda}}=\pi_0|_{\calA_{\Lambda}}$. Now, using $\pi_0 = \pi_\Lambda \otimes \pi_{\Lambda^c}$ so $\pi_0 |_{\calA_{\Lambda^c}} = 1_{\calH_\Lambda} \otimes \pi_{\Lambda^c}$ and $\pi_0 |_{\calA_{\Lambda}} = \pi_\Lambda \otimes 1_{\calH_{\Lambda^c}}$ we have $\Ad(U V_{\sigma,\Lambda}^*)\circ (1_{\calH_\Lambda} \otimes \pi_{\Lambda^c})=(1_{\calH_\Lambda} \otimes \pi_{\Lambda^c})$ and $\Ad(U V_{\sigma,\Lambda^c}^*)\circ (\pi_\Lambda \otimes 1_{\calH_{\Lambda^c}})=(\pi_\Lambda \otimes 1_{\calH_{\Lambda^c}})$.
    Therefore, $U V_{\sigma,\Lambda}^* \in (1_{\calH_\Lambda} \otimes \pi_{\Lambda^c}(\calA_{\Lambda^c}))'=(\calB(\calH_\Lambda)\otimes 1_{\calH_{\Lambda^c}})$.
    Let $V_\Lambda \in \calU(\calH_\Lambda)$ be the unitary such that $U V_{\sigma,\Lambda}^* = V_\Lambda \otimes 1_{\calH_{\Lambda^c}}$ and let $V_{\Lambda^c}\in \calU(\calH_{\Lambda^c})$ be the unitary such that $U V_{\sigma,\Lambda^c}=1_{\calH_\Lambda}\otimes V_{\Lambda^c}$.

    Then, $(V_\Lambda^* \otimes V_{\Lambda^c})=(V_\Lambda \otimes 1_{\calH_{\Lambda^c}})^* \cdot (1_{\calH_\Lambda}\otimes V_{\Lambda^c})=(U V_{\sigma,\Lambda}^*)^* (U V_{\sigma,\Lambda^c}) = V_{\sigma,\Lambda} V_{\sigma,\Lambda^c}^*$.

    As $V_{\sigma,\Lambda} V_{\sigma,\Lambda^c}^* : \calH_{\pi_0}\to \calH_{\pi_0}$ is a composition of two $G$-equivariant maps, it is also equivariant.
    
    At this point, we wish to apply Lemma~\ref{prodOfDisjSuppAndIdGradeImplHomogFactors} to $(V_\Lambda \otimes 1_{\calH_{\Lambda^c}})^* \cdot (1_{\calH_\Lambda}\otimes V_{\Lambda^c})$ having grade $\hat{1}$. Apply the refining of the $\widehat{G}$ grading in Subsection~\ref{SectionRefiningGHatGrading} where the region $\Lambda \subset \Gamma$ chosen is the cone $\Lambda$. As described in Remark~\ref{RemarkGByGCovariantRepFromProductOfGCovReps}, for $\tilde{U}^{(\pi_0)}_{(g_1,g_2)} := U^{(\pi_\Lambda)}_{g_1}\otimes U^{(\pi_{\Lambda^c})}_{g_2}$, $(\pi_0, \tilde{U}^{(\pi_0)})$ is a $G\times G$-covariant representation of $(\calA,\tilde{\alpha})$ such that $\forall g \in G, \tilde{U}^{(\pi_0)}_{(g,g)}=U^{(\pi_0)}_g$. In addition, as $\pi_0 = \pi_\Lambda \otimes \pi_{\Lambda^c}$, $\pi_0(\calA_{\Lambda})''=\pi_\Lambda(\calA_\Lambda)'' \otimes 1_{\calH_{\Lambda^c}}$ and $\pi_0(\calA_{\Lambda^c})'' = 1_{\calH_{\Lambda}}\otimes \pi_{\Lambda^c}(\calA_{\Lambda^c})''$, and so for any non-zero $X \in \pi_0(\calA_{\Lambda})''$ and non-zero $Y \in \pi_0(\calA_{\Lambda^c})''$ we have $XY \ne 0$. Therefore, the conditions of Lemma~\ref{prodOfDisjSuppAndIdGradeImplHomogFactors} are satisfied, so for $A=(V_\Lambda \otimes 1_{\calH_{\Lambda^c}})^*$ and $B=(1_{\calH_\Lambda}\otimes V_{\Lambda^c})$, and $AB=V_{\sigma,\Lambda} V_{\sigma,\Lambda^c}^*$ being $G$-equivariant, i.e. having grade $\hat{1}\in \widehat{G}$, we conclude that $A=(V_\Lambda \otimes 1_{\calH_{\Lambda^c}})^*$ and $B=(1_{\calH_\Lambda}\otimes V_{\Lambda^c})$ are each homogeneous with respect to the $\widehat{G}$-grading, with grades inverses of each-other. Say $\phi$ is the grade of $(1_{\calH_\Lambda}\otimes V_{\Lambda^c})$, so $\phi^{-1}$ is the grade of $(V_\Lambda \otimes 1_{\calH_{\Lambda^c}})^*$.

    For all $g \in G$, $\Ad(U^{(\pi_0)}_g)((V_\Lambda \otimes 1_{\calH_{\Lambda^c}})^*)=\phi^{-1}(g)(V_\Lambda \otimes 1_{\calH_{\Lambda^c}})^*$, so $$\Ad(U^{(\pi_0)}_g)((V_\Lambda \otimes 1_{\calH_{\Lambda^c}}))=(\phi^{-1}(g)(V_\Lambda \otimes 1_{\calH_{\Lambda^c}})^*)^* = \phi(g) (V_\Lambda \otimes 1_{\calH_{\Lambda^c}}),$$ so $(V_\Lambda \otimes 1_{\calH_{\Lambda^c}})$ is homogeneous of grade $\phi$ as well.

    So, with $(V_\Lambda \otimes 1_{\calH_{\Lambda^c}}) = (U V_{\sigma,\Lambda}^*)$ and $(1_{\calH_\Lambda}\otimes V_{\Lambda^c})=U V_{\sigma,\Lambda^c}$ both homogeneous of grade $\phi$, multiplying either by $V_{\sigma,\Lambda}$ or $V_{\sigma,\Lambda^c}$ respectively on the right, we get that $U$ is homogeneous of grade $\phi$ as well (by Lemma~\ref{ProductOfHomogOpsGradesAdd}), because $V_{\sigma,\Lambda}$ and $V_{\sigma,\Lambda^c}$ are $G$-equivariant, i.e. of grade $\hat{1}$, and $\phi \hat{1} = \phi$.

    Finally, if for some unitary $U_2 : \calH_\sigma \to \calH_{\pi_0}$ is satisfies $\Ad(U_2)\circ \sigma = \pi_0$, then, because $\pi_0 = \Ad(U) \circ \sigma$, we have $\Ad(U^* U_2) \circ \sigma = \sigma$, and so $U^* U_2 \in \sigma(\calA)'$, and therefore because $\sigma$ is irreducible, $U^* U_2\in \sigma(\calA)'=\bc 1_{\calH_\sigma}$ and so $U_2$ is just $U$ multiplied by a phase factor, and so has the same grade $\phi$.

    Hence, the irreducible $G$-covariant representations of $(\calA,\alpha)$ satisfying the $G$-equivariant version of the superselection criterion (relative to $(\pi_0, U^{(\pi_0)}_\xbullet)$), are classified by $\widehat{G}$ up to $G$-equivariant unitary equivalence.
\end{proof}

In particular, the result applies to the GNS representation of a pure $G$-invariant product state:
\begin{corollary}\label{CorClassificationForProductState}
Let $\omega = \omega_\Lambda \otimes \omega_{\Lambda^c}$, where 
$\omega_\Lambda$ and $\omega_{\Lambda^c}$ are pure $G$-invariant states on 
$\calA_{\Lambda}$ and $\calA_{\Lambda^c}$ respectively.
Let $(\calH_{\pi_\omega}, \pi_\omega, \Omega_\omega)$ be the GNS representation of $\omega$,
and let $U^{(\pi_\omega)}$ be the unitary $G$-action on $\calH_{\pi_\omega}$ that fixes $\Omega_\omega$
and makes $(\pi_\omega, U^{(\pi_\omega)})$ a $G$-covariant representation of $(\calA,\alpha)$ as in Lemma~\ref{GNSRepOfGInvStateIsGCompatible}.
Then the irreducible $G$-covariant representations that satisfy the
$G$-equivariant superselection criterion with respect to $(\pi_\omega, U^{(\pi_\omega)})$
are classified by $\widehat{G}$ up to $G$-equivariant unitary equivalence.
\end{corollary}
\begin{proof}
Let $(\calH_\Lambda, \pi_\Lambda, \Omega_{\omega_\Lambda})$ and 
$(\calH_{\Lambda^c}, \pi_{\Lambda^c}, \Omega_{\omega_{\Lambda^c}})$ 
be the GNS representations of $\omega_\Lambda$ and $\omega_{\Lambda^c}$.
By Lemma~\ref{GNSRepOfGInvStateIsGCompatible}, there exist 
$U^{(\pi_\Lambda)} : G \to \calU(\calH_\Lambda)$ and 
$U^{(\pi_{\Lambda^c})} : G \to \calU(\calH_{\Lambda^c})$
making $(\pi_\Lambda, U^{(\pi_\Lambda)})$ and 
$(\pi_{\Lambda^c}, U^{(\pi_{\Lambda^c})})$ $G$-covariant representations of 
$(\calA_\Lambda, \alpha_\Lambda)$ and $(\calA_{\Lambda^c}, \alpha_{\Lambda^c})$, respectively, and such that $U^{(\pi_\Lambda)}, U^{(\pi_{\Lambda^c})}$ fix $\Omega_{\omega_\Lambda}, \Omega_{\omega_{\Lambda^c}}$ respectively.

Define
    \[
        \calH_{\pi_0} := \calH_\Lambda \otimes \calH_{\Lambda^c}, \quad
        \pi_0 := \pi_\Lambda \otimes \pi_{\Lambda^c}, \quad
        \Omega_0 := \Omega_{\omega_\Lambda} \otimes \Omega_{\omega_{\Lambda^c}},
    \]
    and set $U^{(\pi_0)}_g := U^{(\pi_\Lambda)}_g \otimes U^{(\pi_{\Lambda^c})}_g$.
Then $(\calH_{\pi_0}, \pi_0, \Omega_0)$ is a GNS representation of $\omega$, and $(\pi_0, U^{(\pi_0)})$ is a $G$-covariant representation of $(\calA,\alpha)$ with $U^{(\pi_0)}$ fixing $\Omega_0$.

By uniqueness of GNS representations up to unitary equivalence, there exists a unitary $V : \calH_{\pi_0} \to \calH_{\pi_\omega}$ such that
$\Ad(V) \circ \pi_0 = \pi_\omega$ and $V \Omega_0 = \Omega_\omega$.
Since $(\pi_\omega, U^{(\pi_\omega)})$ and $(\pi_0, U^{(\pi_0)})$ are both $G$-covariant representations of $(\calA,\alpha)$,
for any $g \in G$ and $A \in \calA$ we have
\[
(U^{(\pi_\omega)}_g)^* V U^{(\pi_0)}_g \pi_0(A)\Omega_0
= (U^{(\pi_\omega)}_g)^* \pi_\omega(\alpha_g(A)) \Omega_\omega
= \pi_\omega(A)\Omega_\omega
= V\pi_0(A)\Omega_0.
\]
Thus $(U^{(\pi_\omega)}_g)^* V U^{(\pi_0)}_g$ and $V$ agree on the dense set 
$\pi_0(\calA)\Omega_0$, and hence coincide on all of $\calH_{\pi_0}$. 
Therefore $V$ is a $G$-equivariant unitary equivalence (Definition~\ref{def:GEquivariantMap}).

Therefore, any irreducible $G$-covariant representation $(\calH_\sigma,\sigma)$ of $(\calA,\alpha)$ that satisfies the $G$-equivariant superselection criterion with respect to $(\pi_\omega,U^{(\pi_\omega)})$ also satisfies it with respect to $(\pi_0,U^{(\pi_0)})$.
Therefore, by Theorem~\ref{classificationOfSectorsForGEquivariantSSCWrtProductRep}, there exists a unique $\phi \in \widehat{G}$ such that there is a unitary $U : \calH_\sigma \to \calH_{\pi_0}$ of grade $\phi$ such that $\Ad(U)\circ \sigma = \pi_0$.
Then $VU : \calH_\sigma \to \calH_{\pi_\omega}$ is a unitary of grade $\phi$ satisfying 
$\Ad(VU)\circ\sigma = \pi_\omega$.
This unitary is unique up to phase, as is standard for intertwiners between 
irreducible representations. 
%because for any unitary $W : \calH_\sigma \to \calH_{\pi_\omega}$ such that $\Ad(W) \circ \sigma = \pi_\omega$, $W^* VU \in \sigma(\calA)'$, and as $\sigma$ is irreducible this must be a phase, and so $W$ must be $VU$ multiplied by a phase.
Hence the irreducible 
$G$-covariant representations satisfying the 
$G$-equivariant superselection criterion with respect to 
$(\pi_\omega,U^{(\pi_\omega)})$ are classified by $\widehat{G}$ up to 
$G$-equivariant unitary equivalence.

\end{proof}

Hence, for $G$-invariant product states, the $G$-equivariant superselection sectors are classified (up to $G$-equivariant unitary equivalence) by the Pontryagin dual $\widehat{G}$.
A natural next step is to explore how aspects of this classification persist or are modified for reference representations that are not product representations and that already exhibit a non-trivial sector structure prior to imposing $G$-equivariance.
It would also be interesting to study the non-abelian case, where one might expect the superselection sectors to be labeled by the finite-dimensional irreducible representations of $G$.

\bibliographystyle{alpha}
\bibliography{bibliography.bib}

\end{document}